\documentclass[letterpaper, 10 pt, conference]{ieeeconf}  % Comment this line out if you need a4paper

\pdfoutput=1

\IEEEoverridecommandlockouts                              % This command is only needed if 
\usepackage[pdftex]{graphicx}
\graphicspath{{./fig/}}
\usepackage[usenames,dvipsnames]{color}
\usepackage{epstopdf}
\usepackage{tabularx}
\usepackage{colortbl}
\usepackage{wrapfig}
\usepackage{stfloats}
\usepackage{booktabs}
\usepackage{hyperref}
\usepackage{cite}
\usepackage [normalem] {ulem}

\definecolor{steelblue}{RGB}{70,130,180}

\usepackage{amsmath,amssymb,mathrsfs,dsfont}
\usepackage{algorithm}
\usepackage{algorithmic}
\usepackage{todonotes}
\newtheorem{example}{Example}

\newtheorem{assumption}{Assumption}
\newtheorem{remark}{Remark}

\newtheorem{theorem}{Theorem}
\newtheorem{corollary}{Corollary}
\newtheorem{lemma}{Lemma}

\DeclareMathOperator{\vect}{vec}

\DeclareMathOperator{\Let}{: =}

\allowdisplaybreaks[4]

\title{\LARGE \bf
Linear System Identification Under Multiplicative Noise \\from Multiple Trajectory Data
}
\author{Yu Xing$^{1}$, Ben Gravell$^{2}$, Xingkang He$^{3}$, Karl Henrik Johansson$^{3}$, and Tyler Summers$^{2}$% <-this % stops a space
\thanks{$^{1}$Y. Xing is with Key Lab of Systems and Control, Academy of Mathematics and Systems Science, Chinese Academy of Sciences, and School of Mathematical Sciences, University of Chinese Academy of Sciences, Beijing, P. R. China, {\tt\small yxing@amss.ac.cn}. His work was supported by National Key R\&D Program of China (2016YFB0901900) and National Natural Science Foundation of China (61573345).}
\thanks{$^{2}$B. Gravell and T. Summers are with Department of Mechanical
Engineering, The University of Texas at Dallas, Richardson, TX, USA, {\tt\small \{Benjamin.Gravell,Tyler.Summers\}@utdallas.edu}. Their work was supported by the Air Force Office of Scientific Research under award number FA2386-19-1-4073.}
\thanks{$^{3}$X. He and K. H. Johansson are with Division of Decision and Control Systems, School of Electrical Engineering and Computer Science, KTH Royal Institute of Technology, Stockholm, Sweden, {\tt\small \{xingkang,  kallej\}@kth.se}. Their work was supported by Knut \& Alice Wallenberg Foundation, and Swedish Research Council.}
}

\begin{document}

\maketitle
\thispagestyle{empty}
\pagestyle{empty}

%%%%%%%%%%%%%%%%%%%%%%%%%%%%%%%%%%%%%%%%%%%%%%%%%%%%%%%%%%%%%%%%%%%%%%%%%%%%%%%%
\begin{abstract}
The study of multiplicative noise models has a long history in control theory but is re-emerging in the context of complex networked systems and systems with learning-based control. We consider linear system identification with multiplicative noise from multiple state-input trajectory data. We propose exploratory input signals along with a least-squares algorithm to simultaneously estimate nominal system parameters and multiplicative noise covariance matrices.
Identifiability of the covariance structure and asymptotic consistency of the least-squares estimator are demonstrated by analyzing first and second moment dynamics of the system. The results are illustrated by numerical simulations.
\end{abstract}

%%%%%%%%%%%%%%%%%%%%%%%%%%%%%%%%%%%%%%%%%%%%%%%%%%%%%%%%%%%%%%%%%%%%%%%%%%%%%%%%
\section{Introduction}

The study of stochastic systems with noise which multiplies with the state and input i.e. multiplicative noise has a long history in control theory \cite{wonham1967optimal}, %\cite{kushner1967stochastic,kozin1969survey}
but is re-emerging in the context of complex networked systems and systems with learning-based control. 
In contrast with the well-known additive noise setting, multiplicative noise has the ability to capture dependence of the noise on the state and/or control input. This situation occurs in modern control systems as diverse as robotics with distance-dependent sensor errors \cite{dutoit2011robot}, networked systems with noisy communication channels \cite{antsaklis2007special,hespanha2007survey}, modern power networks with high penetration of intermittent renewables \cite{guo2019a}, %\cite{carrasco2006power,milano2018foundations}
turbulent fluid flow \cite{lumley2007stochastic}, and neuronal brain networks \cite{breakspear2017dynamic}. Linear systems with multiplicative noise are particularly attractive as a stochastic modeling framework because they remain simple enough to admit closed-form expressions for stability and stabilization via generalized Lyapunov equations (e.g. \cite{boyd1994linear}), optimal control via the solution of generalized Riccati equations \cite{wonham1967optimal,kleinman1969optimal} and state estimation. Additionally, recent results show that the optimal control of this class of systems can be learned strictly from sample data without constructing a model via the reinforcement learning technique of policy gradient \cite{gravell2019learning}. As a complementary perspective, here we tackle the problem from a model-based perspective where the goal is to learn and construct a model from sample data, which can then be used e.g. for optimal control design. 

The first issue that must be addressed is that a complete multiplicative noise system model requires accurate estimates not only of the nominal linear system matrices, but also the noise covariance structure. This stands in stark contrast to the additive noise case where the noise covariance structure has no bearing on the control design and can thus be ignored during system identification.
For the identification of a nominal linear system, recursive algorithms have been developed in the control literature, such as the recursive least-squares algorithm \cite{chen2012identification}. These can be utilized for linear systems with multiplicative noise provided that certain assumptions on the noise and on system stability hold.
For the estimation of noise covariances, both recursive and batch estimation methods have been proposed over the last few decades (see \cite{dunik2017noise} for a review), but these focus nearly exclusively on additive noise. 
In order to estimate multiplicative noise covariances, the maximum-likelihood approach was introduced in \cite{schon2011system}, % shumway2017time}, 
and the Bayesian framework was utilized in \cite{kitagawa1998self}, for example assuming Gaussian or known distributions with unknown parameters. These methods, however, require prior assumptions on the noise distributions whose incorrectness may worsen the performance of the concerned algorithms for optimal control.
Our paper concentrates on jointly estimating the nominal system parameters and the multiplicative noise covariances without imposing any prior assumptions on the distribution of the noises, other than being independent and identically distributed (i.i.d.) with finite first and second moments, which complicates the problem. Both state- and control-dependent noise in the system leads to coupling, which also makes the identification task more difficult. 

The second issue we address is that of performing system identification based on multiple state-input trajectory data rather than a single trajectory.
Multiple trajectory data arises in two broad situations: 1) episodic tasks where a single system is reset to an initial state after a finite run time, as encountered in iterative learning control and reinforcement learning problems \cite{matni2019from} and 2) collecting data from multiple identical systems in parallel, for example, physical experiments \cite{gu2017deep} %, ljung2010perspectives} 
and snapshots of social interaction processes \cite{schaub2019blind}.
For multiple trajectory data the duration of each trajectory sample may be small, but a large sample size can be obtained by virtue of repetition in the case of episodic tasks and parallel execution in the case of multiple identical systems.
Thus, there is a growing interest in system identification based on multiple trajectory data, along with their applications in machine learning literature \cite{dean2017sample,matni2019tutorial}.
%, while a single-trajectory data has been more frequently considered in traditional problems like adaptive control, where time goes to infinity for only one sample path \cite{matni2019from}. 

In this paper we consider linear system identification with multiplicative noise from multiple trajectory data. Our contributions are two-fold:
\begin{itemize}
\item We propose a least-squares estimation algorithm to jointly estimate the nominal system matrices and multiplicative noise covariances from sample averages of multiple finite-horizon trajectory rollouts (Algorithm 1). 
A two-stage algorithm based on first and second moment dynamics that separate the nominal parameters from the noise variances is utilized, where a stochastic input design, from Gaussian and Wishart distributions, is used for exciting the moment dynamics. 
The algorithm does not need prior knowledge for the multiplicative noise or stability conditions for the system, except that the noises are i.i.d. among different trajectories, with finite first and second moments so it may be applied to a wide range of scenarios.
\item Identifiability of the noise covariance matrices and asymptotic consistency of our proposed algorithm are demonstrated.
First, it is shown that there exists an equivalent class of covariance structure that generates the same second moment dynamics. Then, it is verified that dynamics defined by the first and second moments of states can generate a well-defined closed-form expression of the parameters, provided sufficiently exciting input sequences and certain controllability conditions hold. 
Then by assuming the multiple trajectory data are i.i.d., the consistency of the estimator, i.e., convergence to the true value as the number of trajectory samples grows to infinity, is obtained by combining the former result and the law of large numbers.
\end{itemize}

The remainder of the paper is organized as follows:
we formulate the problem in Section \ref{problemFormulation},
then in Section \ref{parameterEstimation} the algorithm is introduced and theoretical results are given,
numerical simulation results are presented in Section \ref{numericalSimulations}, 
and in Section \ref{conclusions} we conclude.

\noindent\textbf{\emph{Notation.}}
We denote the $n$-dimensional Euclidean space by $\mathds{R}^n$, and the set of $n\times m$ real matrices by $\mathds{R}^{n\times m}$. 
We use $\|\cdot\|$ to denote the Euclidean norm for vectors and the Frobenius norm for matrices.
The expectation of a random vector $X$ is represented by $\mathds{E} \{X\}$. Denote the $n$-dimensional identity matrix by $I_n \in \mathds{R}^{n\times n}$.
The Kronecker product of two matrices $A \in \mathds{R}^{m\times n}$ and $B \in \mathds{R}^{p\times q}$ is represented by $A \otimes B$,
and the vectorization of $A$ by ${\vect(A) = (a_{11}~a_{21}~\cdots~a_{m1}~a_{12}~a_{22}~\cdots~a_{mn})^\intercal}$. 
For a block matrix 
\begin{align*}
    B =
    \begin{bmatrix}
    B_{11} & B_{12} & \cdots & B_{1n}\\
    \vdots & \vdots &        & \vdots\\
    B_{m1} & B_{m2} & \cdots & B_{mn}
    \end{bmatrix}
    \in \mathds{R}^{mp \times nq},
\end{align*}
where $B_{ij} \in \mathds{R}^{p\times q}$, we define the following matrix reshaping operator $ F: \mathds{R}^{mp \times nq} \rightarrow \mathds{R}^{mn \times pq}$:
\begin{align*}
    F(B, m, n, p, q) := 
    &[\vect(B_{11}) ~ \vect(B_{21}) ~ \cdots ~ \vect(B_{m1}) ~ \cdots \\
    & \vect(B_{12}) ~ \vect(B_{22}) ~ \cdots ~ \vect(B_{mn})]^\intercal . 
\end{align*}
Then we have that $F(A \otimes A, m, n, m, n) = \vect(A) \vect(A)^\intercal$ for $A \in \mathds{R}^{m \times n}$, which demonstrates the relation between the entries of $A \otimes A$ and those of $\vect(A) \vect(A)^\intercal$. Note when $p=q=1$, $F(\cdot)$ degenerates to $\vect(\cdot)$. 
One can also verify the below reshaping operator from $\vect(A) \vect(A)^\intercal$ to $A \otimes A$. That is, $G: \mathds{R}^{mn \times pq} \to \mathds{R}^{mp \times nq}$:
\begin{align*}
    &G(B, m, n, p, q) := \\
    &\begin{bmatrix}
    \vect_{p\times q}^{-1}(B_1) & \cdots & \vect_{p\times q}^{-1}(B_{(n-1)m+1}) \\
    \vect_{p\times q}^{-1}(B_2) & \cdots & \vect_{p\times q}^{-1}(B_{(n-1)m+2}) \\
    \vdots & & \vdots \\
    \vect_{p\times q}^{-1}(B_m) & \cdots & \vect_{p\times q}^{-1}(B_{mn})
    \end{bmatrix},
\end{align*}
where $B = [B_1^\intercal ~ \cdots B_{mn}^\intercal]^\intercal$ and $\vect_{p\times q}^{-1}(x) = (\vect(I_q)^\intercal \otimes I_p)(I_q \otimes x)$ for $x \in \mathds{R}^{pq}$.

\section{Problem Formulation}\label{problemFormulation}
We consider linear systems with multiplicative noise
\begin{equation}
\begin{aligned}\label{theSystem}
 x_{t+1} &= (A  + \bar{A}_t) x_t +  (B + \bar{B}_t) u_t 
\end{aligned}
\end{equation}
where  $x_t \in \mathds{R}^n$ is the system state and $u_t \in \mathds{R}^m$ is the control input to be designed. 
The dynamics are described by a nominal dynamics matrix $A \in \mathds{R}^{n \times n}$ and nominal input matrix $B \in \mathds{}{R}^{n \times m}$ and incorporate multiplicative noise terms modeled by the i.i.d. and mutually independent random matrices $\bar{A}_t$ and $\bar{B}_t$ which have zero mean and covariance matrices $\Sigma_A := \mathds{E} \{\vect(\bar{A}_t)\vect(\bar{A}_t)^T\} \in \mathds{R}^{n^2 \times n^2}$ and  $\Sigma_B := \mathds{E} \{\vect(\bar{B}_t)\vect(\bar{B}_t)^T\} \in \mathds{R}^{nm \times nm}$. Note that if $\bar{A}_t$ has non-zero mean $\bar{A}$, then we can consider a system with nominal matrix $(A+\bar{A}, B)$, as well as noise terms $\bar{A}_t - \bar{A}$ and $\bar{B}_t$, which satisfies the above zero-mean assumption. This also holds for cases with $\bar{B}_t$ non-zero mean. The term multiplicative noise refers to the fact that noises $\bar{A}_t$ and $\bar{B}_t$ enter the system as multipliers of $x_t$ and $u_t$, rather than as additions. In the latter case, the noises are called additive ones, resulting in much simpler system dynamics.

As an example of system \eqref{theSystem}, consider the following system studied in the optimal control literature \cite{boyd1994linear,gravell2019learning}.
\begin{equation} \label{eq:system_eigen_noises}
\begin{aligned} 
 x_{t+1} &= (A  + \sum_{i = 1}^r A_i p_{i,t}) x_t + (B + \sum_{j = 1}^s B_j q_{j, t}) u_t,
\end{aligned}
\end{equation}
where $\{p_{i,t}\}$ and $\{q_{i,t}\}$ are mutually independent i.i.d. scalar random variables, with $\mathds{E}\{ p_{i,t}\} = \mathds{E} \{q_{j,t}\} = 0$, $\mathds{E} \{p_{i,t}^2\} = \sigma^2_i$, and $\mathds{E}\{ q_{j,t}^2\} = \delta^2_j$, $\forall i \in [1,r], j \in [1,s], t \ge 0$. It can be seen that $\bar{A}_t = \sum_{i = 1}^r A_i p_{i,t}$ and $\bar{B}_t = \sum_{j = 1}^s B_j q_{j, t}$ where $\sigma_i$ and $\delta_j$ are the eigenvalues of $\Sigma_A$ and $\Sigma_B$, and $A_i$ and $B_j$ are the reshaped eigenvectors of $\Sigma_A$ and $\Sigma_B$. These parameters are necessary for optimal controller design, as \cite{gravell2019learning} showed. However, for new systems with unknown parameters, the key problem is to identify them in the first place, stated as follows. Another example of system \eqref{theSystem} is interconnected systems, where the nominal part captures relations among different subsystems, and multiplicative noises characterize randomly varying topologies \cite{haber2014subspace}.
%Moreover, $\{\bar{A}_t\}$ and $\{\bar{B}_t\}$ are independent, so the correlation matrix of $\vect(\bar{A}_t~\bar{B}_t)$, $\Sigma_{AB} := \mathds{E} \vect([\bar{A}_t \ \bar{B}_t]) \vect([\bar{A}_t \ \bar{B}_t])^\intercal \in \mathds{R}^{n(n+m) \times n(n+m)}$ has off-diagonal blocks of all zeros and can be split into the covariance matrices $\Sigma_A$ and  $\Sigma_B$.

\textbf{Problem.} Suppose that the system parameters $A, B, \Sigma_{A},$ and $\Sigma_B$ are unknown, but state-input trajectories are available for system identification.
Our goal in this paper is to estimate $A, B, \Sigma_{A},$ and $\Sigma_B$ based on multiple trajectory data $\{x_t^{(k)}, 0 \le t \le \ell, k \in \mathds{N}^+\}$, by appropriately designing the input sequence $\{u_t^{(k)}, 0 \le t \le \ell-1, k \in \mathds{N}^+\}$ and initial states $x_0^{(k)}$, where $\{ x_t^{(k)} , 0 \le t \le \ell \}$, is the $k$-th trajectory sample, and $\ell$ is the final time-step for every trajectory. 

%\section{System Identification via Least Squares with First and Second Moments}\label{parameterEstimation}

\section{Least-Squares Algorithm Based on Multiple Trajectory Data}\label{parameterEstimation}

\subsection{Algorithm Design} 
In this section, we propose our exploratory input sequence design and least-squares algorithm to estimate the system parameters from multiple trajectory data. 
We assume that the sampled trajectory data are collected independently, and refer to each trajectory sample as a 
\emph{rollout}.
Because every rollout is affected by the multiplicative noise, we will use least-squares on the first and second moment dynamics averaged over multiple trajectories to solve the system identification problem. Also, we assume inputs of arbitrary magnitude may be executed perfectly. 

Taking the expectation of both sides of \eqref{theSystem} we obtain the first-moment dynamics of states, i.e., the dynamics of $\mathds{E}\{x_t\}$,
\begin{align}\label{eq:expectationDynamics}
    \mu_{t+1} = A \mu_t + B \nu_t,
\end{align}
where $\mu_t := \mathds{E}\{x_t\}$ and $\nu_t := \mathds{E} \{u_t\}$. 

Likewise, denote the vectorization of the instantaneous second moment matrices of state, state-input, and input at time $t$ by $X_t := \vect(\mathds{E} \{x_t x_t^\intercal\})$, $W_t := \vect(\mathds{E} \{x_t u_t^\intercal\})$, $W_t' := \vect(\mathds{E} \{u_t x_t^\intercal\})$, and $U_t := \vect(\mathds{E} \{u_t u_t^\intercal\})$. Note that the second moment matrix we used here, namely $\mathds{E}\{ XY^\intercal \}$ for two random vectors $X$ and $Y$, is different from the covariance matrix, which is $\mathds{E}\{ (X-\mathds{E}\{X\}) (Y-\mathds{E}\{Y\})^\intercal \} = \mathds{E}\{ XY^\intercal \} - \mathds{E}\{X\}\mathds{E}\{Y\}^\intercal$.

From the independence of $\bar{A}_t$ and $\bar{B}_t$, as well as vectorization, the second moment dynamics of \eqref{theSystem} are
\begin{align}\label{eq:vecCorrelationDynamics}\nonumber
    X_{t+1} &= (A \otimes A) X_t + (B \otimes A) W_t + (A \otimes B) W_t' \\\nonumber
    &\qquad + (B \otimes B) U_t + \mathds{E} \left\{(\bar{A}_t \otimes \bar{A}_t) \vect(x_t x_t^\intercal) \right\} \\\nonumber
    &\qquad + \mathds{E} \left\{(\bar{B}_t \otimes \bar{B}_t) \vect(u_t u_t^\intercal) \right\}\\\nonumber
    %&= (A \otimes A) X_t + (B \otimes A) W_t + (A \otimes B) W_t^\intercal \\\nonumber
    %&~~ + (B \otimes B) U_t + \Sigma_A' X_t + \Sigma_B' U_t, \\\nonumber
    &= (A \otimes A + \Sigma_A') X_t + (B \otimes B + \Sigma_B') U_t  \\
    &\qquad + (B \otimes A) W_t + (A \otimes B) W_t'
\end{align}
where we denote $\Sigma_A' = \mathds{E}\{\bar{A}_t \otimes \bar{A}_t\} \in \mathds{R}^{n^2 \times n^2}$ and $\Sigma_B' = \mathds{E}\{\bar{B}_t \otimes \bar{B}_t\} \in \mathds{R}^{n^2 \times m^2}$. The relation between $(\Sigma_A, \Sigma_B)$ and $(\Sigma_A', \Sigma_B')$ can be illustrated by $F(\Sigma_A', n, n, n, n) = \Sigma_A$ and $F(\Sigma_B', n, m, n, m) = \Sigma_B$, where the reshaping operator $F(\cdot)$ is defined in the notation subsection.

Before giving an estimation algorithm, it is necessary to talk more about the second moment dynamics \eqref{eq:vecCorrelationDynamics}. Since $\mathds{E}\{x_tx_t^\intercal\}$ is symmetric, $X_t$ has $n(n-1)/2$ pairs of identical entries, corresponding to the off-diagonal entries of $\mathds{E}\{x_tx_t^\intercal\}$. For example, $\mathds{E}\{x_{t,i} x_{t,j}\}$ and $\mathds{E}\{x_{t,j}x_{t,i}\}$. Similarly, $U_t$ has $m(m-1)/2$ pairs of identical entries corresponding to the off-diagonal entries of $\mathds{E}\{u_tu_t^\intercal\}$. 

We apply the following process to \eqref{eq:vecCorrelationDynamics}, to obtain a simplified version of it.

\begin{itemize}
    \item[1.] If there exist $i$ and $j$ such that  $X_{t,i} = \mathds{E}\{x_{t,k}x_{t,l}\}$ and $X_{t,j} = \mathds{E}\{x_{t,l}x_{t,k}\}$, $1 \le i < j \le n^2$, for some $1 \le l < k \le n$, then add the $j$-th column of $A\otimes A + \Sigma_A'$ to its $i$-th column. After that, remove the $j$-th column.
    \item[2.] Remove the $j$-th entry of $X_{t+1}$ and $X_t$, and also the $j$-th row of $A\otimes A + \Sigma_A'$, $B\otimes B + \Sigma_B'$, $A\otimes B$, and $B\otimes A$.
    \item[3.] Return to step 1, until $n(n-1)/2$ entries of $X_t$ are removed (corresponding to the upper-diagonal entries of $\mathds{E}\{x_tx_t^\intercal\}$).
    \item[4.] If there exist $i$ and $j$ such that  $U_{t,i} = \mathds{E}\{u_{t,k}u_{t,l}\}$ and $U_{t,j} = \mathds{E}\{u_{t,l}u_{t,k}\}$, $1 \le i < j \le m^2$, for some $1 \le l < k \le m$, then add the $j$-th column of $B\otimes B + \Sigma_B'$ to its $i$-th column, remove the $j$-th column, and remove the $j$-th entry of $U_t$.
    \item[5.] Return to step 4, until $m(m-1)/2$ entries of $U_t$ are removed (corresponding to the upper-diagonal entries of $\mathds{E}\{u_tu_t^\intercal\}$).
\end{itemize}

In this way, we get a simplified version of \eqref{eq:vecCorrelationDynamics},
\begin{align}\label{eq:vecCorrelationDynamics_simp}\nonumber
    \tilde{X}_{t+1} &= (\tilde{A} + \tilde{\Sigma}_A') \tilde{X}_t + (\tilde{B} + \tilde{\Sigma}_B') \tilde{U}_t \\
    &\qquad + K_{BA} W_t + K_{AB} W_t',
\end{align}
where $\tilde{X}_t \in \mathds{R}^{n(n+1)/2}$ and $\tilde{U}_t \in \mathds{R}^{m(m+1)/2}$. 

The above procedure can be written in a compact form by considering a linear transformation from $\mathds{R}^{n^2}$ to $\mathds{R}^{n(n+1)/2}$, as given below. Let $I_{n^2} := [e_1~\cdots~e_{n^2}]$ be the $n^2$-dimensional identity matrix. Then define matrix $P_1 \in \mathds{R}^{[n(n+1)/2] \times n^2}$ by removing the $[(j-1)n+i]$-th row of $I_{n^2}$, define matrix $T_1 \in \mathds{R}^{n^2 \times n^2}$ by replacing the $[(j-1)n+i]$-th row of $I_{n^2}$ by $e_{(i-1)n+j}^\intercal$, and define matrix $Q_1 \in \mathds{R}^{n^2 \times [n(n+1)/2]}$ by removing the $[(j-1)n+i]$-th column of $T_1$, where $1 \le i < j \le n$. We can also define $P_2 \in \mathds{R}^{[m(m+1)/2] \times m^2}$, $T_2 \in \mathds{R}^{m^2 \times m^2}$, and $Q_2 \in \mathds{R}^{m^2 \times [m(m+1)/2]}$ by changing $n$ to $m$ in the above text.

It can be obtained that $\tilde{X}_t = P_1 X_t$, $\tilde{U}_t = P_2 U_t$, $X_t = T_1 X_t$, $U_t = T_2 U_t$, $T_1 = Q_1P_1$, and $T_2 = Q_2P_2$, by noticing that $\mathds{E}\{x_{t,i}x_{t,j}\}$ is the $[(j-1)n+i]$-th entry of $X_t$, $1\le i,j\le n$, and $\mathds{E}\{u_{t,i}u_{t,j}\}$ is the $[(j-1)m+i]$-th entry of $U_t$, $1\le i,j\le m$.

Hence, from \eqref{eq:vecCorrelationDynamics},
\begin{align*}
    \tilde{X}_{t+1} &= P_1 X_{t+1} \\
    &= P_1 (A \otimes A + \Sigma_A') X_t + P_1 (B \otimes B + \Sigma_B') U_t  \\
    &\qquad + P_1 (B \otimes A) W_t + P_1 (A \otimes B) W_t'\\
    &= P_1 (A \otimes A + \Sigma_A') T_1 X_t + P_1 (B \otimes B + \Sigma_B') T_2 U_t  \\
    &\qquad + P_1 (B \otimes A) W_t + P_1 (A \otimes B) W_t'\\
    &= P_1 (A \otimes A + \Sigma_A') Q_1 P_1 X_t \\
    &\qquad + P_1 (B \otimes B + \Sigma_B') Q_2 P_2 U_t  \\
    &\qquad\qquad  + P_1 (B \otimes A) W_t + P_1 (A \otimes B) W_t'\\
    &= (\tilde{A} + \tilde{\Sigma}_A') \tilde{X}_t + (\tilde{B} + \tilde{\Sigma}_B') \tilde{U}_t \\
    &\qquad + K_{BA} W_t + K_{AB} W_t'.
\end{align*}
As a result, $\tilde{A} = P_1 (A \otimes A) Q_1$, $\tilde{\Sigma}_A' = P_1 \Sigma_A' Q_1$, $\tilde{B} = P_1 (B \otimes B) Q_2$, $\tilde{\Sigma}_B' = P_1 \Sigma_B' Q_2$, $K_{BA} = P_1 (B \otimes A)$, and $K_{AB} = P_1 (A \otimes B)$.

Now with the following fact, we can restate the above relation between $(\Sigma_A',\Sigma_B')$ and $(\tilde{\Sigma}_A',\tilde{\Sigma}_B')$ from an entry-wise perspective in Theorem \ref{thm:identifiability}.

\begin{figure*}[ht]
\begin{align}\label{matrix:sigmaA}
\begin{array}{lc}
\mbox{}&
\begin{array}{ccccc} &(k-1)n+l& &(l-1)n+k \end{array}\\
\begin{array}{c} ~\\ (i-1)n+j \\ ~\\ (j-1)n+i\\~ \end{array}&
\left[\begin{array}{ccccc}
 & \vdots & & \vdots & \\
\cdots & \mathds{E}\{\bar{A}_{t,ik}\bar{A}_{t,jl}\} & \cdots & \mathds{E}\{\bar{A}_{t,il}\bar{A}_{t,jk}\} &\cdots\\
 & \vdots & & \vdots & \\
 \cdots & \mathds{E}\{\bar{A}_{t,jk}\bar{A}_{t,il}\} & \cdots & \mathds{E}\{\bar{A}_{t,jl}\bar{A}_{t,ik}\} & \cdots\\
  & \vdots & & \vdots & 
\end{array}\right]
\end{array}\\\label{matrix:sigmaB}
\begin{array}{lc}
\mbox{}&
\begin{array}{ccccc} &(k-1)m+l& &(l-1)m+k \end{array}\\
\begin{array}{c} ~\\ (i-1)n+j \\ ~\\ (j-1)n+i\\~ \end{array}&
\left[\begin{array}{ccccc}
 & \vdots & & \vdots & \\
\cdots & \mathds{E}\{\bar{B}_{t,ik}\bar{B}_{t,jl}\} & \cdots & \mathds{E}\{\bar{B}_{t,il}\bar{B}_{t,jk}\} &\cdots\\
 & \vdots & & \vdots & \\
 \cdots & \mathds{E}\{\bar{B}_{t,jk}\bar{B}_{t,il}\} & \cdots & \mathds{E}\{\bar{B}_{t,jl}\bar{B}_{t,ik}\} & \cdots\\
  & \vdots & & \vdots & 
\end{array}\right]
\end{array}
\end{align}
\end{figure*}

\begin{lemma} $\Sigma_A'$ and $\Sigma_B'$  has structure in \eqref{matrix:sigmaA} and \eqref{matrix:sigmaB} respectively, where $\bar{A}_{t,ij}$ is the $(i,j)$-th entry of $\bar{A}_t$, $1\le i\not=j \le n$, $1\le k\not=l \le n$, and $\bar{B}_{t,ij}$ is the $(i,j)$-th entry of $\bar{B}_t$, $1\le i\not=j \le n$, $1\le k\not=l \le m$.
\end{lemma}

\begin{proof}
The conclusions follow directly from properties of Kronecker products.
\end{proof}

\begin{theorem}\label{thm:identifiability}
The entries of $\tilde{\Sigma}_A'$ consists of:\\
$\mathds{E}\{\bar{A}_{t,ik}\bar{A}_{t,ik}\},~ 1\le i,k \le n$;\\
$\mathds{E}\{\bar{A}_{t,ik}\bar{A}_{t,jk}\},~ i < j,~ 1\le i,j,k \le n$;\\
$2\mathds{E}\{\bar{A}_{t,ik}\bar{A}_{t,il}\},~ k < l,~ 1\le i,k,l \le n$;\\
$\mathds{E}\{\bar{A}_{t,ik}\bar{A}_{t,jl}\} + \mathds{E}\{\bar{A}_{t,il}\bar{A}_{t,jk}\},~ 1\le i < j \le n, 1\le k < l \le n$.\\ 
The entries of $\tilde{\Sigma}_B'$ consists of:\\
$\mathds{E}\{\bar{B}_{t,ik}\bar{B}_{t,ik}\},~ 1\le i \le n,~ 1\le k \le m$;\\
$\mathds{E}\{\bar{B}_{t,ik}\bar{B}_{t,jk}\}, ~1\le i< j\le n,~ 1\le k\le m$;\\
$2\mathds{E}\{\bar{B}_{t,ik}\bar{B}_{t,il}\},~ 1\le i\le n,~ 1\le k<l \le m$;\\
$\mathds{E}\{\bar{B}_{t,ik}\bar{B}_{t,jl}\} + \mathds{E}\{\bar{B}_{t,il}\bar{B}_{t,jk}\}, ~1\le i < j \le n, 1\le k < l \le m$.
\end{theorem}

\begin{proof}
Consider the simplifying procedure from \eqref{eq:vecCorrelationDynamics} to \eqref{eq:vecCorrelationDynamics_simp}. Step 1 generates columns of additions with form $\mathds{E}\{\bar{A}_{t,ik}\bar{A}_{t,jl}\} + \mathds{E}\{\bar{A}_{t,il}\bar{A}_{t,jk}\}$ for $1\le k < l \le n$. Moreover, if $i = j$, then the corresponding term becomes $2\mathds{E}\{\bar{A}_{t,ik}\bar{A}_{t,il}\}$ . By removing the $[(l-1)n+k]$-th row of $\Sigma_A'$ in step 2, $1\le k < l \le n$, these entries,  $\mathds{E}\{\bar{A}_{t,jk}\bar{A}_{t,il}\} + \mathds{E}\{\bar{A}_{t,jl}\bar{A}_{t,ik}\}, 1\le i < j \le n$ are removed, for all $k < l$. This step also deletes the column consisting of  $\mathds{E}\{\bar{A}_{t,il}\bar{A}_{t,jk}\}$ and $  \mathds{E}\{\bar{A}_{t,jl}\bar{A}_{t,ik}\})$, $1\le i,j \le n$. The entries $\mathds{E}\{\bar{A}_{t,ik}\bar{A}_{t,ik}\}$, $1\le i,k \le n$, remains the same during the process. Same argument for results of $\tilde{\Sigma}_B'$.
\end{proof}

\begin{remark}
The above discussion indicates that $X_t$ is in fact determined by parameter matrices $(A,B)$ and $(\tilde{\Sigma}_A',\tilde{\Sigma}_B')$. It also shows that there exists a set of covariance matrices $(\Sigma_A', \Sigma_B')$ that are equivalent, i.e., $S_{\Sigma}^+ := \{(\Sigma_1', \Sigma_2') \in \mathds{R}^{n^2 \times (n^2 + nm)} : P_1 \Sigma_1' Q_1 = \tilde{\Sigma}_A',P_1 \Sigma_2' Q_2 = \tilde{\Sigma}_B', F(\Sigma_1',n,n,n,n) \succeq 0, F(\Sigma_2',n,m,n,m) \succeq 0\}$, in the sense that they generate an identical second moment dynamic of \eqref{theSystem}. The last two terms in the set definition are due to the positive semi-definiteness of $\Sigma_A$ and $\Sigma_B$. Obviously, $S_{\Sigma}^+$ is not empty for $(\tilde{\Sigma}_A',\tilde{\Sigma}_B')$ obtained in \eqref{eq:vecCorrelationDynamics_simp}, since $(\Sigma_A', \Sigma_B') \in S_{\Sigma}^+$. From an entry-wise point of view, because of multiplicative noises, we cannot recover the exact value of $\mathds{E}\{\bar{A}_{t,ik}\bar{A}_{t,jl}\}$ and $\mathds{E}\{\bar{A}_{t,il}\bar{A}_{t,jk}\}$, if $i\not = j$ and $k \not = l$. One may only estimate the sum of these two entries instead out of $X_t$. Fortunately, some entries of $\Sigma_A'$ and $\Sigma_B'$ are identifiable, e.g., $\mathds{E}\{\bar{A}_{t,ik}\bar{A}_{t,ik}\}$, which is the variance of $\bar{A}_{t,ik}$. This means that if we introduce more conditions for the covariance structure, then it can be obtained exactly. For example, if entries in $\bar{A}_t$ are mutually independent, then $\Sigma_A$ is diagonal, and all of its entries except $\mathds{E}\{\bar{A}_{t,ik}\bar{A}_{t,ik}\}$, $1 \le i,k \le n$, are zero.
\end{remark}

\begin{example}\label{exam:identifiability}
Consider \eqref{theSystem} with $n=2$ and $m=1$, where
$X_t = \left[\mathds{E}\{ x_{t,1} x_{t,1}\} ~ \mathds{E}\{ x_{t,2} x_{t,1}\} ~ \mathds{E}\{ x_{t,1} x_{t,2}\} ~ \mathds{E}\{ x_{t,2} x_{t,2}\}\right]^T$. So $\mathds{E}\{ X_{t,2} X_{t,1}\}$ and $\mathds{E}\{ X_{t,1} X_{t,2}\}$ are identical and have the same update rule in \eqref{eq:vecCorrelationDynamics}. Under this situation, 
\begin{align*}
&\tilde{X}_{t} = \left[\mathds{E}\{ x_{t,1} x_{t,1}\} ~~ \mathds{E}\{ x_{t,2} x_{t,1}\} ~~ \mathds{E}\{ x_{t,2} x_{t,2}\}\right]^T,\\
&P_1 = \begin{bmatrix}
1 & 0 & 0 & 0\\
0 & 1 & 0 & 0\\
0 & 0 & 0 & 1
\end{bmatrix}, ~
Q_1 = \begin{bmatrix}
1 & 0 & 0 \\
0 & 1 & 0 \\
0 & 1 & 0 \\
0 & 0 & 1
\end{bmatrix}, 
\\
&T_1 = \begin{bmatrix}
1 & 0 & 0 & 0 \\
0 & 1 & 0 & 0 \\
0 & 1 & 0 & 0 \\
0 & 0 & 0 & 1 
\end{bmatrix}, P_2 = Q_2 = T_2 = 1.
\end{align*}
According to the above simplification, from 
\begin{align*}
\Sigma_A' &= \begin{bmatrix}
\sigma_{a, 11,11} & \sigma_{a, 11,12} & \sigma_{a, 12,11} & \sigma_{a, 12,12}\\
\sigma_{a, 11,21} & \sigma_{a, 11,22} & \sigma_{a, 12,21} & \sigma_{a, 12,22}\\
\sigma_{a, 21,11} & \sigma_{a, 21,12} & \sigma_{a, 22,11} & \sigma_{a, 22,12}\\
\sigma_{a, 21,21} & \sigma_{a, 21,22} & \sigma_{a, 22,21} & \sigma_{a, 22,22}
\end{bmatrix},\\
\Sigma_B' &= \begin{bmatrix}
\sigma_{b, 11} & \sigma_{b, 12} & \sigma_{b,21} & \sigma_{b,22}
\end{bmatrix}^T,
\end{align*}
we have
\begin{align*}
\tilde{\Sigma}_A' &= \begin{bmatrix}
\sigma_{a, 11,11} & \sigma_{a, 11,12} + \sigma_{a, 12,11} & \sigma_{a, 12,12}\\
\sigma_{a, 11,21} & \sigma_{a, 11,22} + \sigma_{a, 12,21} & \sigma_{a, 12,22}\\
\sigma_{a, 21,21} & \sigma_{a, 21,22} + \sigma_{a, 22,21} & \sigma_{a, 22,22}
\end{bmatrix},\\
\tilde{\Sigma}_B' &= \begin{bmatrix}
\sigma_{b, 11} & \sigma_{b, 12} & \sigma_{b,22}
\end{bmatrix}^T,
\end{align*}
where $\sigma_{a, ij,kl} = \mathds{E}\{\bar{A}_{t,ij} \bar{A}_{t,kl}\}$, $\sigma_{b, ij} = \mathds{E}\{\bar{B}_{t,i}, \bar{B}_{t,j}\}$, and \begin{align*}
&\tilde{A} = \begin{bmatrix}
a_{11}a_{11} & a_{11}a_{12} + a_{12}a_{11} & a_{12}a_{12}\\
a_{11}a_{21} & a_{11}a_{22} + a_{12}a_{21} & a_{12}a_{22}\\
a_{21}a_{21} & a_{21}a_{22} + a_{22}a_{21} & a_{22}a_{22}
\end{bmatrix}, \\
&\tilde{B} = \begin{bmatrix}
b_1b_1 & b_1b_2  & b_2b_2
\end{bmatrix} ^T, \\
&K_{AB} = \begin{bmatrix}
a_{11}b_1 & a_{12}b_1\\
a_{21}b_1 & a_{22}b_2\\
a_{21}b_2 & a_{22}b_2
\end{bmatrix}, \quad
K_{BA} = \begin{bmatrix}
a_{11}b_1 & a_{12}b_1\\
a_{11}b_2 & a_{12}b_2\\
a_{21}b_2 & a_{22}b_2
\end{bmatrix}.
\end{align*}
\end{example}

The first and second moment dynamics \eqref{eq:expectationDynamics} and \eqref{eq:vecCorrelationDynamics_simp} are linear in the dynamic model parameters to be estimated. It is natural to consider a two-stage least-squares procedure, where first the nominal system matrices ($A$,$B$) are estimated from \eqref{eq:expectationDynamics}, and then these estimates are plugged in to obtain estimates for the variances ($\Sigma_A$, $\Sigma_B$) from \eqref{eq:vecCorrelationDynamics_simp}. If we had access to the exact first and second moment dynamics, this procedure would produce exact estimates. % {\color{black} \sout{provided a sufficient rollout length and persistent excitation of both first \emph{and} second moment dynamics from the input mean $\nu_t$ \emph{and} second moment $U_t$}}. 
%Yu: Under the condition of knowing the exact first and second moment dynamics, we do not need any design and excitation? (because we can use least-squares directly)
However, we must estimate the first and second moment dynamics from rollout data, and we propose to take a sample average over multiple independent rollouts. 
To obtain persistently exciting inputs, we randomly generate the first and second moment of the input sequence %{\color{black}\sout{distribution}} 
from standard Gaussian and Wishart\footnote{The Wishart distribution $W_p(V,n)$ is the probability distribution of the matrix $X = G G^\intercal$ where each column of the matrix $G$ is drawn from the $p$-variate Gaussian distribution $\mathcal{N}_p(0,V)$. Clearly Wishart distributions are supported on the set of positive semidefinite matrices.} distributions \cite{gupta2018matrix}, respectively. Likewise, the initial states are assumed to be randomly drawn from a distribution $\mathcal{X}$ with finite second moment (see Sec. \ref{subsubsec:consistency}). The overall algorithm is shown in Algorithm \ref{alg:A}, where the superscript $(k)$ represents the $k$-th rollout. 

% \begin{remark}
% Since ${\Sigma}_A$ and ${\Sigma}_B$ are the covariance matrices of $\bar{A}_t$ and $\bar{B}_t$ they must be positive semidefinite. Hence a positive semidefinite constraint must be imposed on the optimization problem in line 15 of Algorithm \ref{alg:A} which can be easily achieved by generic convex optimization parser-solvers such as \textup{CVX} in \textup{MATLAB} \cite{cvx}. %, Grant2008}. 
% Note that $\Sigma_A^\prime$ and $\Sigma_B^\prime$ are related to $\Sigma_A$ and $\Sigma_B$ via one-to-one mappings by inverse of the $F(\cdot)$ operator. However, if the estimator is consistent (as we will prove later it is), then as the amount of sample data increases the estimated covariances will become arbitrarily close to the true values and the semidefinite constraint will naturally become ineffective.
% \end{remark}

In Alg. \ref{alg:A} and in the sequel, to ease notation, we omit the "$\tilde{~}$" above $X_t$, $U_t$, $\Sigma_A'$, and $\Sigma_B'$, but readers should keep in mind that they are the simplified version of their counterparts in \eqref{eq:vecCorrelationDynamics}.

\begin{algorithm}[t]
% \caption{Least-squares system identification for linear systems with multiplicative noise from multiple trajectory data}
\caption{\newline Multiple-trajectory averaging least-squares (MALS)}
% \caption{Multiple trajectory averaging least-squares (MTrAveLS)}
% \caption{Multiple trajectory averaging least-squares (MultiTrAveLS)}
\label{alg:A}
\begin{algorithmic}[1]
\FOR{$t$ from $0$ to $\ell$}
\STATE{Generate $\nu_t$ and $\bar{U}_t$ independently from zero-mean Gaussian and Wishart \cite{gupta2018matrix} distributions, respectively. Both $\nu_t$ and $\bar{U}_t$ are fixed after generation}
\ENDFOR
\FOR{$k$ from $1$ to $n_r$}
\STATE{Generate $x_{0}^{(k)}$ independently from the distribution $\mathcal{X}$} \\
\FOR{$t$ from $0$ to $\ell-1$}
\STATE{Generate $u_t^{(k)}$ independently from the Gaussian distribution $\mathcal{N}(\nu_t, \bar{U}_t)$}
\STATE{$x_{t+1}^{(k)} = (A + \bar{A}_t^{(k)}) x_t^{(k)} + (B + \bar{B}_t^{(k)}) u_t^{(k)}$}
\ENDFOR
\ENDFOR
\FOR{$t$ from $0$ to $\ell$}
\STATE{Compute 
\begin{align*}
    \hat{\mu}_t & \Let \frac1{n_r} \sum_{k = 1}^{n_r} x_t^{(k)} , \\ 
    \hat{X}_t   & \Let \frac{1}{n_r} \vect \left( \sum_{k = 1}^{n_r} x_t^{(k)} (x_t^{(k)})^\intercal \right) , \\
    \hat{W}_t   & \Let \frac{1}{n_r} \vect \left( \sum_{k = 1}^{n_r} x_t^{(k)} \nu_t^\intercal \right) = \vect(\hat{\mu}_t \nu_t^\intercal ), \\
    \hat{W}_t'   & \Let \frac{1}{n_r} \vect \left( \sum_{k = 1}^{n_r} \nu_t {x_t^{(k)}}^\intercal \right) = \vect(\nu_t \hat{\mu}_t^\intercal ), \\
         U_t    & \Let \vect(\bar{U}_t + \nu_t\nu_t^\intercal)
\end{align*}
}
\ENDFOR
\STATE{$(\hat{A},\hat{B}) = {\text{argmin}}_{(A, B)} \{ \frac12 \sum_{t=0}^{\ell-1}  \|\hat{\mu}_{t+1} - (A\hat{\mu}_{t}+B \nu_{t}) \|_2^2 \}$}
\STATE{$(\hat{\Sigma}_A',\hat{\Sigma}_B') = %\underset{\Sigma_A^\prime \succeq 0, \Sigma_B^\prime \succeq 0}{\text{argmin}} 
\text{argmin}_{(\Sigma_A', \Sigma_B')} \{ \frac12 \sum_{t=0}^{\ell-1}  \|\hat{X}_{t+1} - [\tilde{A} \hat{X}_t + K_{BA} \hat{W}_t + K_{AB} \hat{W}_t' + \tilde{B} U_t  + \Sigma_A' \hat{X}_t + \Sigma_B' U_t] \|_2^2 \}$}
% \STATE{$\hat{\Sigma}_A = F(\hat{\Sigma}_A', n, n, n, n) $, $\hat{\Sigma}_B = F(\hat{\Sigma}_B', n, m, n, m)$}
\end{algorithmic}
\end{algorithm}

\subsection{Theoretical Consistency Analysis}

In this section we analyze the consistency of Algorithm \ref{alg:A} by investigating the moment dynamics \eqref{eq:expectationDynamics} and \eqref{eq:vecCorrelationDynamics_simp}, which motivated the least-squares approach in Algorithm \ref{alg:A}. 
%Due to space constraints proofs of the results are deferred to supplementary material. TODO: decide arXiv?
%Yu: I uploaded the paper to arXiv, and we can cite it. (Not announced yet..)

\subsubsection{Moment Dynamics}
Note again if we know $\mu_t$ and $X_t$, then it is possible to recover the parameters via least-squares as in lines $14$-$15$ in Algorithm \ref{alg:A}. Let
\begin{equation}
\begin{aligned}\label{defLSMatrices}
    &\mathbf{Y}
    := \begin{bmatrix}
    \mu_{\ell} & \cdots & \mu_1
    \end{bmatrix},\quad
    \mathbf{Z}
    := \begin{bmatrix}
    \mu_{\ell-1} & \cdots & \mu_0\\
    \nu_{\ell-1} & \cdots & \nu_0
    \end{bmatrix}, \\
    &\mathbf{C} :=
    \begin{bmatrix}
    C_\ell & \cdots & C_1
    \end{bmatrix}, \quad
    \mathbf{D} :=
    \begin{bmatrix}
    X_{\ell-1} & \cdots & X_0\\
    U_{\ell-1} & \cdots & U_0
    \end{bmatrix}, 
\end{aligned}
\end{equation}
where $C_t = X_t - [\tilde{A} X_{t-1} + K_{BA} W_{t-1} + K_{AB} W_{t-1}' + \tilde{B} U_{t-1}]$, $1 \le t \le \ell$. Then closed-form solutions of the least-squares problems are
\begin{align*}
    (\hat{A}, \hat{B}) &= \mathbf{Y} \mathbf{Z}^\intercal (\mathbf{Z} \mathbf{Z}^\intercal)^{\dagger},\\
    (\hat{\Sigma}_A', \hat{\Sigma}_B') &= \mathbf{C} \mathbf{D}^\intercal (\mathbf{D} \mathbf{D}^\intercal)^{\dagger},
\end{align*}
where $\mathbf{C}$, $\mathbf{D}$, $\mathbf{Y}$, and $\mathbf{Z}$ are defined in \eqref{defLSMatrices} above, and the sign $\dagger$ represents the pseudoinverse. When the inverse matrices exist, the solutions are identical to true values, that is, $(\hat{A}, \hat{B}) = (A, B)$ and $(\hat{\Sigma}_A', \hat{\Sigma}_B') = (\Sigma_A', \Sigma_B')$.
Hence, the first question towards the consistency of the algorithm is whether the matrices $\mathbf{Z} \mathbf{Z}^\intercal$ and $\mathbf{D} \mathbf{D}^\intercal$ are invertible, which is necessary for the consistency of the algorithm.
As to be shown, this invertibility can be obtained by designing a proper input sequence, if systems $(A, B)$ and $(\tilde{A} + \Sigma_A', \tilde{B} + \Sigma_B')$ are controllable, and the final time-step $\ell$ is large enough. In fact, in this paper we randomly generalized the first and second moments of inputs to ensure the invertibility. As a consequence, we need to demonstrate the following results in a probability sense, intuitively saying that random generation of input statistics results in the expected invertibility.
%The precise statements are presented in Theorems \ref{thm: as_Z_fullrank} and \ref{thm: as_D_fullrank}.

\begin{theorem}\label{thm: as_Z_fullrank}
Suppose that $\ell \ge \frac12 mn^2 + \frac12 mn + m + 1$ and $(A, B)$ is controllable. The matrix $\mathbf{Z}$ has full row rank with probability one, and consequently $\mathbf{Z} \mathbf{Z}^\intercal$ is invertible, if the entries of $\nu_t$, $0 \le t \le \ell-1$, are generated i.i.d. from a non-degenerate Gaussian distributions.
\end{theorem}
\begin{proof}
See Appendix A.
\end{proof}

\begin{remark}
The above theorem shows that for large enough time step of each rollout, the full row rank condition of $\mathbf{Z}$ can be guaranteed with probability one if the mean of the input at each time step is generated randomly and independently. In the proof, the controllability of $(A, B)$ plays a key role. In addition, although the lower bound in the theorem is relatively small, one may conjecture that $\ell \ge n+m$ is a sharp lower bound for the invertibility of $\mathbf{Z}\mathbf{Z}^\intercal$, which will be a future work.
\end{remark}

\begin{theorem}\label{thm: as_D_fullrank}
Suppose that $\ell \ge \frac12 m^2n^4 + \frac12 m^2n^2 + m^2 + 1$ and $(\tilde{A} + \Sigma_A', \tilde{B} + \Sigma_B')$ is controllable. The matrix $\mathbf{D}$ has full row rank with probability one, and consequently $\mathbf{D} \mathbf{D}^\intercal$ is invertible, if $\nu_t$ have been fixed and the entries of $\bar{U}_t$ are generated i.i.d. from a non-degenerate Wishart distributions, $0 \le t \le \ell-1$, where $\bar{U}_t$ is defined in line 2 of Algorithm \ref{alg:A}.
\end{theorem}

\begin{proof}
See Appendix B.
\end{proof}

\begin{remark}
The controllability condition in Theorem \ref{thm: as_D_fullrank} reflects the nature of the multiplicative noise, i.e., coupling between $\bar{A}_t$ and $x_t$, and that between $\bar{B}_t$ and $u_t$. It also indicates that a controllability condition on \eqref{eq:vecCorrelationDynamics}, the dynamics of the second %{\color{black}\sout{-order}} 
moments of states, is necessary to ensure the successful identification of $\Sigma_A'$ and $\Sigma_B'$.
\end{remark}

\begin{corollary}
Suppose that $\ell \ge \frac12 m^2n^4 + \frac12 m^2n^2 + m^2 + 1$, and both $(A, B)$ and $(\tilde{A} + \Sigma_A', \tilde{B} + \Sigma_B')$ are controllable. The matrices $\mathbf{Z}\mathbf{Z}^\intercal$ and $\mathbf{D}\mathbf{D}^\intercal$ are invertible, if first the entries of $\nu_t$ are generated i.i.d. from a non-degenerate Gaussian distribution and then $\bar{U}_t$ is generated i.i.d. from a non-degenerate Wishart distribution, $0 \le t \le \ell-1$, where $\bar{U}_t$ is defined in line 2 of Algorithm \ref{alg:A}.
\end{corollary}

\begin{remark}
From the proof of Theorems \ref{thm: as_Z_fullrank} and \ref{thm: as_D_fullrank}, we know that the existence of the inverses of $\mathbf{Z}\mathbf{Z}^\intercal$ and $\mathbf{D}\mathbf{D}^\intercal$ can in fact be guaranteed with probability one, as long as $\nu_t$ and $\bar{U}_t$, the mean and vectorized second moment matrix of the input at time $t$, are generated independently from a distribution that is absolutely continuous with respect to Lebesgue measure. Also note the random generation of the first and second moments of inputs leads to non-stationarity of the input sequence. Critically this provides sufficient excitation of both the first and second moments of the state and makes it possible to estimate all model parameters in the presence of multiplicative noise.
\end{remark}

\subsubsection{Consistency} \label{subsubsec:consistency}
After the discussion in the previous section, we now assume that the means and second moments of the input sequences have been generated in Algorithm \ref{alg:A}, and both $\mathbf{Z}\mathbf{Z}^\intercal$ and $\mathbf{D}\mathbf{D}^\intercal$ have been designed to be invertible. The closed-form estimates generated by Algorithm \ref{alg:A} are
\begin{align}\label{LSEstimator1}
    (\hat{A}, \hat{B}) &= \hat{\mathbf{Y}} \hat{\mathbf{Z}}^\intercal (\hat{\mathbf{Z}} \hat{\mathbf{Z}}^\intercal)^{\dagger},\\\label{LSEstimator2}
    (\hat{\Sigma}_A', \hat{\Sigma}_B') &= \hat{\mathbf{C}} \hat{\mathbf{D}}^\intercal (\hat{\mathbf{D}} \hat{\mathbf{D}}^\intercal)^{\dagger},
\end{align}
where
\begin{align*}
    &\hat{\mathbf{Y}}
    := \begin{bmatrix}
    \hat{\mu}_{\ell} & \cdots & \hat{\mu}_1
    \end{bmatrix},\quad
    \hat{\mathbf{Z}}
    := \begin{bmatrix}
    \hat{\mu}_{\ell-1} & \cdots & \hat{\mu}_0\\
    \nu_{\ell-1} & \cdots & \nu_0
    \end{bmatrix}, \\
    &\hat{\mathbf{C}} :=
    \begin{bmatrix}
    \hat{C}_\ell & \cdots & \hat{C}_1
    \end{bmatrix}, \quad
    \hat{\mathbf{D}} :=
    \begin{bmatrix}
    \hat{X}_{\ell-1} & \cdots & \hat{X}_0\\
    U_{\ell-1} & \cdots & U_0
    \end{bmatrix}, 
\end{align*}
and $\hat{C}_t = \hat{X}_t - [\hat{\tilde{A}} \hat{X}_{t-1} + \hat{K}_{BA} \hat{W}_{t-1} + \hat{K}_{AB} \hat{W}_{t-1}' + \hat{\tilde{B}} U_{t-1}]$, $1 \le t \le \ell$. Here $\hat{\tilde{A}}$, $\hat{\tilde{B}}$, $\hat{K}_{AB}$, and $\hat{K}_{BA}$ are estimates of $\tilde{A}$, $\tilde{B}$, $K_{AB}$, and $K_{BA}$, obtained by using $\hat{A}$ and $\hat{B}$ from Algorithm \ref{alg:A}. The estimates above depend on the number of rollouts $n_r$, but we omit it for convenience. Before stating the consistency result, we present the following assumption for the system and data:

\begin{assumption}\label{asmp1}
For all rollouts, the below conditions hold. \\
(i) The final time-step is fixed to be $\ell \geq \frac{1}{2} m^2n^4 + \frac{1}{2} m^2n^2 + m^2 + 1$.\\
(ii) The initial state $x_0^{(k)}$, $1 \le k \le n_r$, is generated independently from the same distribution with $\mathds{E}\{\|x_0^{(k)}\|^2\} < \infty$, and is independent of the subsequent process. \\
(iii) $\{\bar{A}_t^{(k)}, 0 \le t \le \ell\}$ and $\{\bar{B}_t^{(k)}, 0 \le t \le \ell\}$, $1 \le k \le n_r$, have zero mean and finite second moments, i.e., $\mathds{E}\{\bar{A}_t\} = \mathds{E}\{\bar{B}_t\} = \mathbf{0}$ and $\|\Sigma_A\|, \|\Sigma_B\| < \infty$. Also, these sequences are i.i.d. and mutually independent.\\%they are mutually independent, and i.i.d. respectively.\\
(iv) The input signals are generated by Line 6 of Algorithm \ref{alg:A} and are such that both $\mathbf{Z}\mathbf{Z}^\intercal$ and $\mathbf{D}\mathbf{D}^\intercal$ are invertible.
\end{assumption}

Under Assumption \ref{asmp1} the rollouts $x_0^{(k)}, \dots, x_l^{(k)}$, $1 \le k \le n_r$, are i.i.d., so consistency can be established from Kolmogorov's strong law of large numbers. %\cite{van2000asymptotic}. 

\begin{theorem}{(Consistency)}\label{thm: consistency}
Suppose that Assumption \ref{asmp1} holds, then the estimators \eqref{LSEstimator1}-\eqref{LSEstimator2} are asymptotically consistent, i.e., 
\begin{align*}
    (\hat{A}, \hat{B}) \to (A, B), \ \ \text{ and } \ \
    (\hat{\Sigma}_A', \hat{\Sigma}_B') \to (\Sigma_A', \Sigma_B'), 
\end{align*}
with probability one as the number of rollouts $n_r \to \infty$.
\end{theorem}

\begin{proof}
See Appendix C.
\end{proof}

\begin{remark}
This theorem indicates that despite the relatively small final time-step for each trajectory, an increasing number of rollouts compensates for this deficiency and guarantees asymptotic estimation performance.
\end{remark}
% \begin{remark}
% The matrices $\hat{\mathbf{Z}}$ and $\hat{\mathbf{D}}$ become full-rank after $n+m$ and $n^2+m^2$ time steps respectively with probability 1 if the multiplicative noise has non-degenerate density with respect to Lebesgue measure of corresponding dimension. To see this, note that the input $u_t$ is a vector of continuous random variables which are statistically independent between all rollouts and time steps. Thus $\bar{B}_t u_t$ is a vector of continuous random variables which are statistically independent of previous time steps and of all rollouts. Since $\hat{\mu}_t$ is the sum of the $x_t$, $\hat{\mu}_t$ is also a vector of continuous random variables which are statistically independent of previous time steps. Therefore all the $\hat{\mu}_t$ are linearly independent and $\hat{\mathbf{Z}}$ becomes full-rank with probability 1 after $n+m$ steps. A similar argument applies for $\hat{\mathbf{D}}$. Nevertheless, we must impose a larger lower bound on $\ell$ from Assumption \ref{asmp1} in order to achieve the consistency guarantee of Theorem \ref{thm: consistency} as the number of rollouts $n_r \rightarrow \infty$.
% \end{remark}

From the estimates of $\tilde{\Sigma}_A'$ and $\tilde{\Sigma}_B'$ (note that in Theorem \ref{thm: consistency} we omit the notation "$\tilde{~}$" for simplicity), explicit forms of covariance structure $\Sigma_A$ and $\Sigma_B$ can be given as follows:
\begin{align*}
    \Sigma_A(\alpha) &= F\left( Q_1 \tilde{\Sigma}_A' Q_3^\intercal + E_{\alpha}, n, n, n, n \right),\\
    \Sigma_B(\beta) &= F\left( Q_1 \tilde{\Sigma}_B' Q_4^\intercal + E_{\beta}, n, m, n, m \right),
\end{align*}
where 
\begin{align*}
    E_{\alpha} &= \underset{1\le k < l \le n}{\underset{1 \le i < j \le n}{\sum}} \bigg[ \alpha_{ij,kl} (e_{(i-1)n+j} - e_{(j-1)n+i})\\
    &\qquad\qquad\qquad \cdot (e_{(k-1)n+l} - e_{(l-1)n+k})^\intercal\bigg],\\
    E_{\beta} &= \underset{1\le k < l \le m}{\underset{1 \le i < j \le n}{\sum}} \bigg[ \beta_{ij,kl} (e_{(i-1)n+j} - e_{(j-1)n+i})\\
    &\qquad\qquad\qquad \cdot (f_{(k-1)m+l} - f_{(l-1)m+k})^\intercal\bigg],
\end{align*}
with $I_{n^2} = [e_1~\cdots~e_{n^2}]$, $I_{m^2} = [f_1~\cdots~f_{m^2}]$, $\alpha_{ij,kl}, \beta_{ij,pq} \in \mathds{R}$, $1\le i< j\le n$, $1\le k< l\le n$, $1\le p < q \le m$, and $Q_1$ is defined after \eqref{eq:vecCorrelationDynamics_simp}. In addition, $Q_3 = D_n Q_1$ and $Q_4 = D_m Q_2$, where $D_n$ is an $n^2$-dimensional diagonal matrix with $[(i-1)n+j]$-th diagonal entry being $1/2$ and the rest being $1$, $1\le i\not= j\le n$, and $D_m$ is an $m^2$-dimensional diagonal matrix with $[(i-1)m+j]$-th diagonal entry being $1/2$ and the rest being $1$, $1\le i\not= j\le m$.

\section{Numerical Simulations}\label{numericalSimulations}

To empirically validate our theoretical consistency result, we simulated our least-squares estimator %\footnote{For computational efficiency we solve the unconstrained least-squares problem for $\hat{\Sigma}_A$ and $\hat{\Sigma}_B$ then project onto the positive semidefinite cone; for all but the smallest numbers of rollouts this projection is ineffective.} 
on two example systems. The first is a simple 2-state, 1-input system where we use a large amount of data to show \textit{asymptotic trends}, while the second is an 8-state, 8-input system representing lossy diffusion dynamics on a network for a more \textit{practical application}. Python code which implements the algorithms and performs the simulated experiments described here is available on GitHub at https://github.com/TSummersLab/sysid-multinoise%\url{https://github.com/TSummersLab/sysid-multinoise}.

\subsection{Simple example}
We consider a simple example system with $n=2$, $m=1$,
${
    A  =
    \begin{bmatrix}
        -0.2 & 0.3 \\
        -0.4 & 0.8
    \end{bmatrix}, 
    B = 
    \begin{bmatrix} 
        -1.8 \\
        -0.8
    \end{bmatrix},
}$
and noise covariances
\begin{align*}
    \Sigma_A & = \frac{1}{100}
    \begin{bmatrix}
         8 & -2 & 0 &  0     \\
        -2 & 16 & 2 &  0     \\
         0 & 2  & 2 &  0     \\
         0 & 0  & 0 &  8   
    \end{bmatrix} ,
    \Sigma_B = \frac{1}{100}
    \begin{bmatrix}
         5 &  -2 \\
        -2 &  20
    \end{bmatrix}.
\end{align*}

\begin{figure}
\begin{center}
\includegraphics[height=6cm]{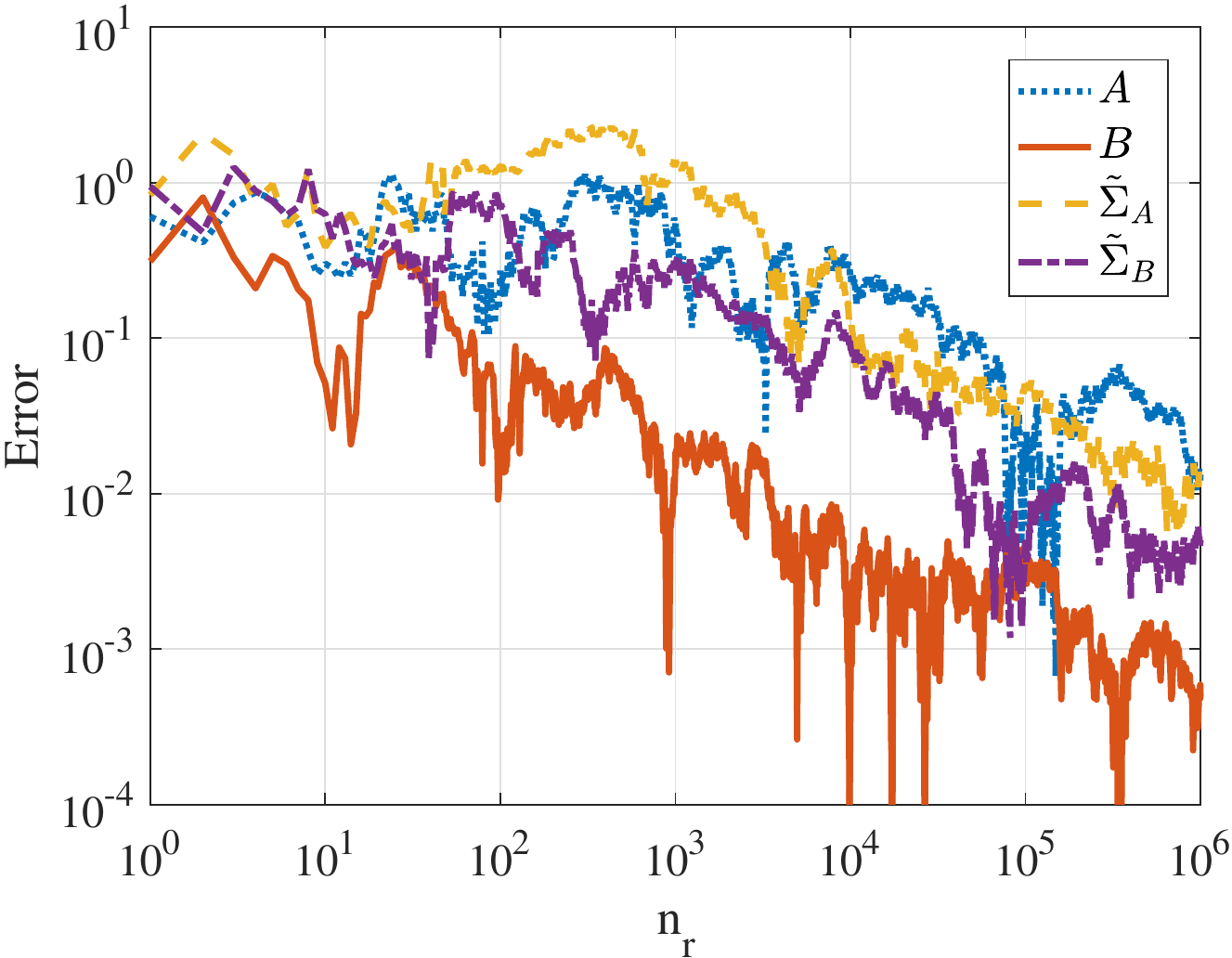}    % The printed column  
\caption{Consistency of Alg. \ref{alg:A}.}  % width is 8.4 cm.
\label{fig:error_convergence}                                 % Size the figures 
\end{center}                                 % accordingly.
\end{figure}

According to the reshaping operator $G$ defined in the notation subsection and Example \ref{exam:identifiability}, we have
\begin{align*}
    \tilde{\Sigma}_A' & = \frac{1}{100}
    \begin{bmatrix}
         8 & 0 & 2      \\
        -2 & 2 & 0     \\
         16 & 0 & 8  
    \end{bmatrix} ,
    \tilde{\Sigma}_B' = \frac{1}{100}
    \begin{bmatrix}
         5 &  -2 & 20
    \end{bmatrix}^\intercal.
\end{align*}

We performed a simulated experiment where rollout data of length ${\ell = \frac{1}{2} m^2n^4 + \frac{1}{2} m^2n^2 + m^2 + 1 = 12}$. 
We used control inputs distributed as $u_t \sim \mathcal{N}(\nu_t,\bar{U}_t)$, where $\nu_t$ and $\bar{U}_t$ are generated from $\mathcal{N}(0, I_n )$ and $\mathcal{W}_n(0.1 I_n,n)$, respectively, and then are fixed. 
Model estimates were computed at 100 increasing logarithmically spaced numbers of rollouts between $1$ and $n_{r}$. The result is plotted in Fig. \ref{fig:error_convergence}, indicating the consistency of the proposed algorithm.

% It is well known that least-squares estimation of finite-impulse response (FIR) models yields estimates whose error decreases as $\mathcal{O}(N^{-1/2})$ where $N$ is the number of samples \cite{Ljung1986}. Empirically we observed a similar trend on the convergence rate; we conjecture that $\hat{A}$ and $\hat{B}$ converge to $A$ and $B$ as $\mathcal{O}(n_r^{-1/2})$ while $\hat{\Sigma}_A$ and $\hat{\Sigma}_B$ converge to $\Sigma_A$ and $\Sigma_B$ at a slower rate, perhaps $\mathcal{O}(n_r^{-\frac{1}{2n}})$ and $\mathcal{O}(n_r^{-\frac{1}{2m}})$ respectively. This can be observed in the similarity between the dashed red reference curves and the median empirical model estimate curves. However it is difficult, even for small example systems like this one, to deduce the asymptotic convergence rate from empirical observation.

% Although our consistency result implies that theoretically the estimates converge regardless of the size of the variance of the random inputs, in practice the design of the inputs has a large impact on the (transient) quality of the estimates. In particular, it is important to strike a balance between exciting the system modes and multiplicative noises and not overwhelming the state- and input-dependent noises by the (random) control inputs i.e. the magnitude of the control input means $\nu_t$ and covariances $\bar{U}_t$ should be chosen in a ``sweet spot''. We used values that gave good results empirically for the example system considered.

%%%%%%%%%%%%%%%%%%%%%%%%%%%%%%%%%%
% START: IN-WORK MATERIAL

\subsection{Network example}

Many practical networked systems can be approximated by diffusion dynamics with loss; examples include heat flow through uninsulated pipes, hydraulic flow through leaky pipes, information flow between processors with packet loss, electrical power flow between generators with resistant electrical power lines, etc. These dynamics in continuous-time act on an undirected graph with no self-loops with symmetric weighted adjacency matrix $A_c$, degree matrix $D_c = \text{diag}(A_c \textbf{1}_{n \times 1})$, graph Laplacian $L = D_c - A_c$, diagonal loss matrix $F_c$, and diagonal input matrix $B_c$:
\begin{align}
    \dot{x} = -(L_c + F_c) x + B_u u
\end{align}
Discretizing these dynamics using the forward Euler method with a step size $T$ yields $x_{t+1} = A x_t + B u_t$
% \begin{align}
%     x_{t+1} &= x_t - T (L_c + F_c) x_t + T B_u u_t \\
%             &= (I - T (L_c + F_c)) x_t + (T B_u) u_t \\
%             &= A x_t + B u_t
% \end{align}
where $A = I - T (L_c + F_c)$ and $B = T B_u$. 
Uncertainty on an edge weight of the graph i.e. on entry $(j,k)$ of $A_c$ manifests as a noise matrix with entries
\begin{align*}
    [A_i]_{p,q} = 
    \begin{cases}
    +1 \text{ if } \{ j=p \text{ and } k=p \} \text{ or } \{ j=q \text{ and } k=q \} , \\
    -1 \text{ if } \{ j=q \text{ and } k=p \} \text{ or } \{ j=p \text{ and } k=q \} , \\
    0 \text{ otherwise. }
    \end{cases}
\end{align*}
Uncertainty on an input strength i.e. entry $(k,k)$ of $B_u$ manifests as a noise matrix with entries
\begin{align*}
    [B_j]_{p,q} = 
    \begin{cases}
    +1 \text{ if } k=q=p, \\
    0 \text{ otherwise. }
    \end{cases}    
\end{align*}

For computational tractability we estimated only the noise variances while giving the estimator knowledge of the noise directions $A_i$ and $B_j$. To formulate this setting mathematically it is easier to work with the eigendecomposition of the noises as in \eqref{eq:system_eigen_noises}. The least-squares estimation for this case is a simple modification to the full covariance estimator:
\begin{align}\label{LSEstimator1_var_only}
    (\hat{\sigma}^2, \hat{\delta}^2) &= \hat{\mathbf{C}} \hat{\mathbf{D}}^\intercal (\hat{\mathbf{D}} \hat{\mathbf{D}}^\intercal)^{\dagger},
\end{align}
where
$(\hat{\sigma}^2, \hat{\delta}^2)$ are vectors of the noise variances,
% \begin{align*}
%     \hat{\mathbf{C}} & :=
%     {\vect} \big(
%     \begin{bmatrix}
%     \hat{C}_\ell & \cdots & \hat{C}_1
%     \end{bmatrix}
%     \big), \\
%     \hat{\mathbf{D}} & :=
%     \begin{bmatrix}
%     {\vect} \big( (A_1 \otimes A_1) \hat{X}_{\ell-1} & \cdots & (A_1 \otimes A_1) \hat{X}_0 \big) \\
%     \vdots \\
%     {\vect} \big( (A_r \otimes A_r) \hat{X}_{\ell-1} & \cdots & (A_r \otimes A_r) \hat{X}_0 \big) \\
%     {\vect} \big( (B_1 \otimes B_1) U_{\ell-1} & \cdots & (B_1 \otimes B_1) U_0 \big) \\
%     \vdots \\
%     {\vect} \big( (B_s \otimes B_s) U_{\ell-1} & \cdots & (B_s \otimes B_s) U_0 \big)
%     \end{bmatrix}, 
% \end{align*}
\begin{align*}
    \hat{\mathbf{C}} & :=
    {\vect} \big(
    \begin{bmatrix}
    \hat{C}_\ell & \cdots & \hat{C}_1
    \end{bmatrix}
    \big), \\
    \hat{\mathbf{D}} & :=
    \begin{bmatrix}
    {\vect} \big( \tilde{A}_1 \hat{X}_{\ell-1} & \cdots & \tilde{A}_1 \hat{X}_0 \big) \\
    \vdots \\
    {\vect} \big( \tilde{A}_r \hat{X}_{\ell-1} & \cdots & \tilde{A}_r \hat{X}_0 \big) \\
    {\vect} \big( \tilde{B}_1 U_{\ell-1} & \cdots & \tilde{B}_1 U_0 \big) \\
    \vdots \\
    {\vect} \big( \tilde{B}_s U_{\ell-1} & \cdots & \tilde{B}_s U_0 \big)
    \end{bmatrix}, 
\end{align*}
$\tilde{A}_i$ and $\tilde{B}_j$, $1\le i\le r$, $1\le j\le s$, are defined in the same way as $\tilde{A}$ and $\tilde{B}$ in \eqref{eq:vecCorrelationDynamics_simp} respectively, and $\hat{C}_t = \hat{X}_t - [\hat{\tilde{A}} \hat{X}_{t-1} + \hat{K}_{BA} \hat{W}_{t-1} + \hat{K}_{AB} \hat{W}_{t-1}^\intercal + \hat{\tilde{B}} U_{t-1}]$, $1 \le t \le \ell$. %As before, $\hat{A}$ and $\hat{B}$ are obtained by Algorithm \ref{alg:A}.

We chose a network with $n=8$ nodes and edges placed via the Erdos-Renyi random graph generation with random integer weights. The graph was selected to be connected so that the system would be controllable.
We used rollout data of length ${\ell = \frac{1}{2} m^2n^4 + \frac{1}{2} m^2n^2 + m^2 + 1 = 133185}$ and collected 7 rollouts; more rollouts could be used, but empirically this amount of data was sufficient to give good estimates. 
% The noise variances $\sigma^2_i$ were selected between $0.005$ and $0.020$ and $\delta^2_j$ between $0.05$ and $0.20$. 
Table \ref{tab:table} shows the averages and maximums of the normalized noise variance estimation errors
% $\overline{\sigma}_i^2 = \frac{| \sigma_i^2 - \hat{\sigma}_i^2 |}{\sigma_i^2}$ and $\overline{\delta}_j^2 = \frac{| \delta_j^2 - \hat{\delta}_j^2 |}{\delta_j^2}$
\begin{align*}
    \overline{\sigma}_i^2 = \frac{| \sigma_i^2 - \hat{\sigma}_i^2 |}{\sigma_i^2}, \quad \text{and} \quad \overline{\delta}_j^2 = \frac{| \delta_j^2 - \hat{\delta}_j^2 |}{\delta_j^2} .
\end{align*}

% \begin{figure}[!ht] 
%     \centering
%     \includegraphics[width=0.3\linewidth]{MALS/Images/network.png}
%     \caption{Graph representation of the networked system. Line widths are scaled to the edge weights.}
%     \label{fig:network_graph}
% \end{figure}

\vspace{-5mm}

\begin{table}[!ht] 
\caption{Estimation error averages and maximums.}
% \caption{ }
\centering
\normalsize
\begin{tabular}{ c c c c }
\toprule
  $ \frac{1}{r} \sum_i^r \overline{\sigma}_i^2$  
& $ \text{max}_i  \overline{\sigma}_i^2 $ 
& $ \frac{1}{s} \sum_j^s \overline{\delta}_j^2$  
& $ \text{max}_j  \overline{\delta}_j^2$  \\
\midrule\\
\addlinespace[-3ex]
  0.0358
& 0.132
& 0.0517
& 0.141
\\
\bottomrule
\end{tabular}

\label{tab:table}
\end{table}

% Figure \ref{} shows the entries of the system matrices, noise covariances, their estimates and the estimate error while Table \ref{} shows the Frobenius norm and maximum absolute entry of the error e.g. $A-\widehat{A}$.

% \begin{figure}
%     \centering
%     \includegraphics[width=\linewidth]{MALS/Images/network_estimates_AB.png}
%     \caption{Visual representation of the system matrices, noise covariances, their estimates and the estimate error. }
%     \label{fig:network_estimates_AB}
% \end{figure}

% \begin{figure}
%     \centering
%     \includegraphics[width=\linewidth]{MALS/Images/network_estimates_SigmaAB.png}
%     \caption{Visual representation of the noise covariances, their estimates and the estimate error. }
%     \label{fig:network_estimates_SigmaAB}
% \end{figure}

% END: IN-WORK MATERIAL
%%%%%%%%%%%%%%%%%%%%%%%%%%%%%%%%%%

\section{Conclusions}\label{conclusions}
In this paper we proposed a system identification scheme for linear systems with multiplicative noise based on multiple trajectory data. 
By designing appropriate persistently exciting input signals, a least-squares algorithm was proposed for the joint estimation of nominal system and multiplicative noise covariances.
The asymptotic consistency of the algorithm was proved, and illustrated by numerical simulations.
Ongoing and future research directions include studying the convergence rate and non-asymptotic behavior of the proposed algorithm, problems of optimal input design, identification from single-trajectory data, and sparsity-promoting regularization for identification of networked systems with prior knowledge of sparsity levels.

%%%%%%%%%%%%%%%%%%%%%%%%%%%%%%%%%%%%%%%%%%%%%%%%%%%%%%%%%%%%%%%%%%%%%%%%%%%%%%%%

%%%%%%%%%%%%%%%%%%%%%%%%%%%%%%%%%%%%%%%%%%%%%%%%%%%%%%%%%%%%%%%%%%%%%%%%%%%%%%%%

%\addtolength{\textheight}{-7cm}

%%%%%%%%%%%%%%%%%%%%%%%%%%%%%%%%%%%%%%%%%%%%%%%%%%%%%%%%%%%%%%%%%%%%%%%%%%%%%%%%


\section*{APPENDIX A: The Proof of Theorem \ref{thm: as_Z_fullrank}}    % Each appendix must have a short title.                                        % in the appendices.

We begin by stating a standard result regarding the zero set of a polynomial which will be needed later:
\begin{lemma}\label{lem: polynomial}
A polynomial function $\mathds{R}^n$ to $\mathds{R}$ is either identically $0$ or non-zero almost everywhere.
\end{lemma}
\begin{proof}
The conclusion is a standard result from real analysis \cite{caron2005zero,federer2014geometric}%; we omit the proof due to space limitations.
\end{proof}

Now we provide a lemma which will naturally lead to the conclusion of Theorem \ref{thm: as_Z_fullrank}:
\begin{lemma}\label{zeroset_Z_fullrank}
Consider the moment dynamics \eqref{eq:expectationDynamics}. Suppose that $\ell > \frac12 mn^2 + \frac12 mn + m + 1$ and $(A, B)$ is controllable. The set of $\nu_t$ such that the rank of $\mathbf{Z}$ is less than $n+m$,
\begin{align}
    \mathcal{V} := \{(\nu_0^\intercal~ \cdots~ \nu_{\ell-1}^\intercal)^\intercal \in \mathds{R}^{m\ell}: \textup{rank}(\mathbf{Z}) < n+m\}
\end{align}
is of Lebesgue measure zero.
\end{lemma}

\begin{proof}
Since $\textup{rank}(\mathbf{Z}) < n+m$ if and only if all $(n+m) \times (n+m)$ minors of $\mathbf{Z}$ are $0$, \[\mathcal{V} = \cap_{k = 1}^{C_{\ell}^{n+m}} \{(\nu_0^\intercal~ \cdots~ \nu_{\ell-1}^\intercal)^\intercal \in \mathds{R}^{m\ell}: [\mathbf{Z}]_k = 0\},\] where $\{[\mathbf{Z}]_k$, $1 \le k \le C_{\ell}^{n+m}\}$ contains all $(n+m)$-order minors of $\mathbf{Z}$. 
So it suffices to prove that, under the conditions of the lemma, there exists an $(n+m)$-order minor $[\mathbf{Z}]_k$ such that the set of input expectations $\{(\nu_0^\intercal~ \dots~ \nu_{\ell-1}^\intercal)^\intercal \in \mathds{R}^{m\ell}: [\mathbf{Z}]_k = 0\}$ is of Lebesgue measure zero.

From the definition of $\mathbf{Z}$, $(n+m)$-order minors of $\mathbf{Z}$ are polynomials of  $(\nu_0^\intercal~ \dots~ \nu_{\ell-1}^\intercal)^\intercal$ with coefficients being the entries of matrices $A^{\ell-1}B$, \dots, $B$, as well as $\mu_0$. 
Thus, if we can show that $[\mathbf{Z}]_k$ is not trivial (equal to zero almost surely), then its zero set is of Lebesgue measure zero by Lemma \ref{lem: polynomial}, which implies the conclusion.

It follows from the definition of controllability of $(A, B)$ that there exist $B_{i_1}$, $AB_{i_1}$, $\dots$, $A^{r_1}B_{i_1}$, $\dots$, $B_{i_p}$, $\dots$, $A^{r_p}B_{i_p}$, $1 \le i_1 < \cdots < i_p \le m$, $0 \le r_k \le n-1$ for $1 \le k \le p$, such that they form a basis of $\mathds{R}^n$, where $B_i$ is the $i$-th column of $B$. 
We sort these $n$ vectors according to the ascending order of the power of $A$: $A^{s_0}B_{f_{0}(1)}$, $\dots$, $A^{s_0} B_{f_{0}(q_{0})}$, $A^{s_1}B_{f_{1}(1)}$, $\dots$, $A^{s_1}B_{f_{1}(q_{1})}$, $\dots$, $A^{s_v}B_{f_{v}(q_{v})}$, where $0 = s_0 < s_1 < s_2 < \dots < s_v \le n-1$, $f_{k}(\cdot) \in \{1, \dots, m\}$ is strictly increasing functions, $1 \le q_k \le m$, $0 \le k \le v$, $v(\le n-1)$ is the total number of different power of $A$ appearing in the above basis, and $\sum_{k = 0}^v q_{k} = n$.
%&=\begin{bmatrix}
%    \sum\limits_{t=0}^{l-2} \sum\limits_{j=1}^m a_{1j}^{(t)} u_{l-2-t, j} & \sum\limits_{t=0}^{l-3} \sum\limits_{j=1}^m a_{1j}^{(t)} u_{l-3-t, j} & \dots & \sum\limits_{j=1}^m a_{1j}^{(0)} u_{0, j} & 0\\
%    \sum\limits_{t=0}^{l-2} \sum\limits_{j=1}^m a_{2j}^{(t)} u_{l-2-t, j} & \sum\limits_{t=0}^{l-3} \sum\limits_{j=1}^m a_{2j}^{(t)} u_{l-3-t, j} & \dots & \sum\limits_{j=1}^m a_{2j}^{(0)} u_{0, j} & 0\\
%    \vdots & \vdots & & \vdots & \vdots \\
%    \sum\limits_{t=0}^{l-2} \sum\limits_{j=1}^m a_{nj}^{(t)} u_{l-2-t, j} & \sum\limits_{t=0}^{l-3} \sum\limits_{j=1}^m a_{nj}^{(t)} u_{l-3-t, j} & \dots & \sum\limits_{j=1}^m a_{nj}^{(0)} u_{0, j} & 0\\
%    u_{l-1, 1} & u_{l-2, 1} & \cdots & u_{11} & u_{01} \\
%    \vdots & \vdots & & \vdots & \vdots \\
%    u_{l-1, m} & u_{l-2, m} & \cdots & u_{1m} & u_{0m}
%    \end{bmatrix}
For $1 \le d \le n+m = m + \sum_{k=0}^v q_k$, define
\begin{align*}
    h_d &= 
    \begin{cases}
    d, \qquad\quad\, \text{ for } 1 \le d \le m,\\
    m + r, \quad \text{ for } d = m + r, 1 \le r \le q_0,\\
    m+q_0+1+(r-1)(s_1+1),  \\
    \qquad\qquad \text{ for } d = m+q_0+r, 1 \le r \le q_1,\\
    \dots\\
    m+\sum_{p=0}^{k-1}(q_{p}(s_p+1))+1 + (r - 1)(s_k+1),  \\
    \qquad\qquad \text{ for } d = m+\sum_{p=0}^{k-1}q_p + r, 1 \le r \le q_k,\\
    \dots\\
    m+\sum_{p=0}^{v-1}(q_{p}(s_p+1)) + 1 + (r-1)(s_v+1),  \\
    \qquad\qquad \text{ for } d = m+\sum_{p=0}^{v-1}q_p+r, 1 \le r \le q_v,
    \end{cases}
\end{align*}
where $s_k$, $q_{k}$, and $v$ are defined above. 

Now we write the entries of $\mathbf{Z}$ explicitly in \eqref{Z} and select from left to right the $h_1$-th, $\dots$, $h_{n+m}$-th columns of \eqref{Z}.
\begin{figure*}[ht]
\begin{align}\label{Z}
    \mathbf{Z}
    =\begin{bmatrix}
    A^{\ell-1}\mu_0 + \sum\limits_{t=0}^{\ell-2} A^t B \nu_{\ell-2-t}, & A^{\ell-2} \mu_0 + \sum\limits_{t=0}^{\ell-3} A^t B \nu_{\ell-3-t}, & \cdots, & A \mu_0 + B \nu_0, & \mu_0\\
    \nu_{\ell-1}, & \nu_{\ell-2}, & \cdots, & \nu_1, & \nu_0
    \end{bmatrix}
\end{align}
\begin{align}\label{key_minor}
    &\Bigg|
    \begin{array}{cccc}
    A^{\ell-h_1}\mu_0 + \sum\limits_{t=0}^{\ell-h_1-1} A^t B \nu_{\ell-h_1-1-t},  & \cdots, & A^{\ell-h_{n+m}} \mu_0 + \sum\limits_{t=0}^{\ell-h_{n+m}-1} A^t B \nu_{\ell-h_{n+m}-1-t}\\
    \nu_{\ell-h_1}, &  \cdots, & \nu_{\ell-h_{n+m}}
    \end{array}\Bigg|
\end{align}
\begin{align}\label{key_minor_minor}
    &\left|
    \begin{array}{ccc}
    A^{\ell-h_{m+1}}\mu_0 + \sum\limits_{t=0}^{\ell-h_{m+1}-1} A^t B \nu_{\ell-h_{m+1}-1-t}, & \cdots, & A^{\ell-h_{n+m}} \mu_0 + \sum\limits_{t=0}^{\ell-h_{n+m}-1} A^t B \nu_{\ell-h_{n+m}-1-t}
    \end{array}  \right|
\end{align}
\begin{align}\label{key_minor_ab}
    \left|
    B_{f_{0}(1)},~\cdots,~B_{f_{0}(q_{0})},~A^{s_1}B_{f_{1}(1)},~\cdots,~A^{s_1}B_{f_{1}(q_{1})},~A^{s_2}B_{f_{2}(1)},~\cdots,~A^{s_2}B_{f_{2}(q_{2})},~\cdots,~A^{s_v}B_{f_{v}(q_{v})}
    \right| \not= 0,
\end{align}
\end{figure*}
This can be done because
\begin{align*}
    &\quad~ m+\sum_{p=0}^{v-1}(q_{p}(s_p+1)) + 1 + (q_v-1)(s_v+1)\\
    &\le m+\sum_{p=0}^{v-1}(m(s_p+1)) + 1 + (m-1)(s_v+1)\\
    &\le m+\sum_{p=0}^{n-2}(m(p+1)) + 1 + (m-1)((n-1)+1)\\
    &= \frac12 (mn^2+mn-2n+2m+2) < \ell.
\end{align*}
These $n+m$ columns define an $(n+m)$-order minor of $\mathbf{Z}$ as in \eqref{key_minor}. As said previously, \eqref{key_minor} is a polynomial of the entries of $(\nu_0^\intercal~ \dots~ \nu_{\ell-1}^\intercal)^\intercal$, and moreover it is non-zero almost everywhere. To see this, we analyze the coefficient of the term 
\begin{align}\label{term_u}
\prod_{d=1}^{m} \nu_{\ell-h_d, d} \prod_{d=m+1}^{m+n} \nu_{\ell-h_{d+1}, z_d},
\end{align}
where $h_{m+n+1} := m+\sum_{p=0}^{v}(q_{p}(s_p+1)) + 1$,
\begin{align*}
    z_d = 
    \begin{cases}
    f_0(r), & \text{ for } d = m+r, 1 \le r \le q_0,\\
    \dots\\
    f_k(r), & \text{ for } d = m+\sum_{p=0}^{k-1}q_p +r, 1 \le r \le q_k,\\
    \dots\\
    f_v(r), & \text{ for } d=m+\sum_{p=0}^{v-1}q_p+r, 1 \le r \le q_v,
    \end{cases}
\end{align*}
and $\nu_{i,j}$ is the $j$-th component of $\nu_i$. As above, $h_{m+n+1}$ is well defined because of the assumption for $\ell$: ${h_{m+n+1} = h_{n+m} + s_v + 1 \le h_{n+m} + n \le \ell}$. Now note for $1 \le d \le m$ that $\nu_{l-h_d}$ in \eqref{term_u} no longer appears at columns on its right side in \eqref{key_minor}, which can be observed from \eqref{Z}. 
So the absolute value of the coefficient of \eqref{term_u} in \eqref{key_minor} is determined by the upper right $n\times n$ determinant of \eqref{key_minor}, namely \eqref{key_minor_minor}.

From observing \eqref{Z}, it follows that $\nu_{\ell-h_{d}}$ only appears at the first $h_{d}$ columns of \eqref{Z}, and moreover only appears at the first $h_d-1$ columns of the first $n$ rows in \eqref{Z}, for $m+2 \le d \le m+n+1$. Hence, $\nu_{h_{d+1}, z_d}$ only appears once at the $(d-m)$-th column of \eqref{key_minor_minor}, $m+1 \le d \le m+n$. Also note that the difference of $h_{d+1}$ and $h_d$ for $m+1 \le d \le m+n$ is $h_{d+1} - h_d = s_k+1$ for $d = m+\sum_{p=0}^{k-1}q_p + r$, $1 \le d \le q_k$, and $0 \le k \le v$, so the corresponding term of $\nu_{\ell-h_{d+1}}$ in the summation at the $h_{d}$-th column of \eqref{Z} is $A^{s_k}B\nu_{\ell-h_{d}-1-s_k} = A^{s_k}B\nu_{\ell-h_{d+1}}$. Thus, the absolute value of the coefficient of \eqref{term_u} is identical to the determinant \eqref{key_minor_ab}, from the selection of $\nu$ in \eqref{term_u} and the fact that $A^tB\nu_{i} = \sum_{j=1}^m A^t B_j \nu_{i, j}$, where $B_j$ is the $j$-th column of $B$. Here the columns containing $A^{s_v}$ can be selected because the assumption for $l$ ensures that $h_{n+m+1} = h_{n+m} + s_v + 1 \le \ell$.
Therefore, we show that the polynomial of inputs \eqref{key_minor} is non-zero almost everywhere, and consequently the conclusion of the theorem holds.
\end{proof}

%\addtolength{\textheight}{-3.5cm}

\quad\emph{Proof of Theorem \ref{thm: as_Z_fullrank}:}
The conclusion follows from Lemma \ref{zeroset_Z_fullrank} and the fact that the probability density function of a non-degenerate Gaussian is absolutely continuous with respect to the Lebesgue measure of corresponding dimension.

\section*{APPENDIX B: The Proof of Theorem \ref{thm: as_D_fullrank}}
We write \eqref{eq:vecCorrelationDynamics_simp} as 
\begin{align*}
    X_{t+1} &= (\tilde{A} + \Sigma_A') X_t + (\tilde{B} + \Sigma_B') U_t \\
    &\quad + [K_{BA} W_t + K_{AB} W_t']\\
    &:= \breve{A} X_t + \breve{B} U_t + \eta_t\\
    &= \breve{A}^{t+1} X_0 + \sum_{k = 0}^t \breve{A}^k \breve{B} U_{t-k} + \sum_{k = 0}^t \breve{A}^k \eta_{t-k},
\end{align*}
so the conclusion follows from an argument essentially identical to that of Theorem \ref{thm: as_Z_fullrank} by noticing that $\eta_t$, $0 \le t \le \ell-1$, are fixed, and by considering $U_t$ as an input.

\section*{APPENDIX C: The Proof of Theorem \ref{thm: consistency}}
Consider each rollout $[(x_0^{(k)})^\intercal, \dots, (x_\ell^{(k)})^\intercal]^\intercal$ as an independent sample of the random vector $\mathbf{x}_{\ell} := [x_0^\intercal, \dots, x_\ell^\intercal]^\intercal$, and from Assumption \ref{asmp1} (ii) and (iii) we know that the random vector $\mathbf{x}_{\ell}$ has finite first and second moments.
So it follows from the Kolmogorov's strong law of large numbers that $\hat{\mathbf{Y}} \to \mathbf{Y}$ a.s., and similarly $\hat{\mathbf{Z}} \to \mathbf{Z}$ a.s., as $n_r \to \infty$. From the assumption that $\mathbf{Z} \mathbf{Z}^\intercal$ is invertible and the continuous mapping theorem (Theorem 2.3 in \cite{van2000asymptotic}), it can be obtained that as $n_r \to \infty$
\begin{align*}
    \hat{\mathbf{Y}} \hat{\mathbf{Z}}^\intercal (\hat{\mathbf{Z}} \hat{\mathbf{Z}}^\intercal)^{-1} \to \mathbf{Y} \mathbf{Z}^\intercal (\mathbf{Z} \mathbf{Z}^\intercal)^{-1}, \text{ a.s.}
\end{align*}
Here $(\mathbf{Z} \mathbf{Z}^\intercal)^{-1}$ exists because of assumption. Note that $\hat{\mathbf{Z}}$ is the average of trajectories, which depends on independent Gaussian inputs $u_t^{(k)}$. So From Lemma \ref{lem: polynomial}, $(\hat{\mathbf{Z}} \hat{\mathbf{Z}}^\intercal)^{-1}$ exists with probability one. If it does not exist, then we can define it to be the zero matrix.
Combining the above convergence with the Kolmogorov's strong law of large numbers, the convergence of $\hat{\mathbf{C}}$ and $\hat{\mathbf{D}}$ follows. Therefore, applying the continuous mapping theorem again, we obtain the consistency of the estimator $(\Sigma_A', \Sigma_B')$.

\bibliographystyle{IEEEtran}
\bibliography{IEEEabrv,bibliography.bib}

\end{document}